\documentclass[11pt]{article}
\usepackage[a4paper,margin=1in]{geometry}
\usepackage{amsmath,amsthm,amssymb,mathtools,bm}
\usepackage{bbm}
\usepackage{physics}
\usepackage{hyperref}
\hypersetup{colorlinks=true,linkcolor=blue,citecolor=blue,urlcolor=blue}
\usepackage{microtype}
\usepackage{enumitem}
\usepackage{tikz}    
\usepackage{float}  

\numberwithin{equation}{section}

\newtheorem{definition}{Definition}[section]
\newtheorem{theorem}[definition]{Theorem} 
\newtheorem{proposition}[definition]{Proposition}
\newtheorem{remark}[definition]{Remark}

\begin{document}

\title{Intrinsic Geometry of Operational Contexts:\\
A Riemannian-Style Framework for Quantum Channels}

\author{Kazuyuki Yoshida\\
Independent Researcher, Minoo-City, Osaka, Japan\\
\texttt{kazuyuki.manny.yoshida@gmail.com}}
\date{\today}
\maketitle

\begin{abstract}
We propose an intrinsic geometric framework on the space of operational contexts,
specified by channels, stationary states, and self-preservation functionals.
Each context $C$ carries a pointer algebra, internal charges, and a self-consistent
configuration minimizing a self-preservation functional. The Hessian of this functional
yields an intrinsic metric on charge space, while non-commutative questioning loops
$dN \to d\Phi \to d\rho^\circ$ define a notion of curvature. In suitable regimes,
this N--Q--S geometry reduces to familiar Fisher-type information metrics and
admits charts that resemble Riemannian or Lorentzian space--times. We outline
how gauge symmetries and gravitational dynamics can be interpreted as holonomies
and consistency conditions in this context geometry.
\end{abstract}


\section{Introduction}
In this section we motivate N--Q--S intrinsic geometry as a framework
that treats contexts, channels, and self-preservation as primary,
and regards GR and the Standard Model as external charts on this geometry.

We briefly review the N--Q--S layers (N-layer openness, Q-layer channels,
S-layer self-preservation), and indicate how they can be reorganized
into a geometric language parallel to Riemannian geometry.
From a broader perspective, this work sits at the intersection of several existing lines of research. Information geometry endows statistical manifolds with Fisher-type metrics derived from divergence functionals, providing powerful tools for estimation and learning problems \cite{amariNagaoka2000,petz1996}. Quantum information geometry extends these ideas to monotone metrics on state spaces and to operationally meaningful distinguishability measures. In parallel, operational reconstructions of quantum theory and gravity emphasize channels, reference frames, and entanglement structure as the primary objects, with space–time and classicality emerging only in suitable limits \cite{hardy_op_gr,hoehnwever2018,blancoCasiniHungMyers2013,swingle2012,jacobson2016}. Our starting point is closer in spirit to these operational approaches, but applied directly to the geometry of contexts and self-preservation rather than to kinematical state spaces.

The contributions of this paper are threefold. First, we formalize a context space whose objects carry pointer algebras, internal symmetry groups, and self-preservation functionals, and we define an intrinsic metric as the Hessian of these functionals with respect to internal charges. Second, we introduce a context Laplacian and a global self-preservation action, leading to discrete self-consistency equations that reduce, in coarse-grained limits, to diffusion- and elasticity-type field equations on emergent manifolds. Third, we define sensitivity triples, NQS-connections, and NQS-curvature as operational analogs of tangent vectors, connections, and curvature, and we show how classical Riemannian geometry and Fisher-type information metrics appear as limiting cases. Simple Markov-chain and qutrit examples illustrate how the abstract structures can be computed in practice.

\subsection*{Overview of the N--Q--S construction}
Operationally, an N--Q--S geometry is built in three layers.
The N-layer parametrizes openness and external conditions through a graph of contexts, the Q-layer assigns to each context a quantum dynamical generator with a stationary state, and the S-layer endows each context with a self-preservation functional whose Hessian defines an intrinsic metric on a space of internal charges.
Small changes in openness then propagate along the chain
\[
  dN \;\to\; d\mathcal{L}_C \;\to\; d\rho_C^\circ \;\to\; dq(C) \;\to\; d^2\mathcal{F}_C,
\]
which we will call the \emph{sensitivity chain}.
The rest of the paper makes this schematic picture precise and shows how familiar Riemannian and information-geometric structures arise as limiting cases.

\begin{definition}[N--Q--S architecture]\label{def:NQS-architecture}
An \emph{N--Q--S architecture} consists of
\begin{itemize}[leftmargin=2em]
  \item an N-layer: a directed graph of contexts $\mathcal{C}$ with objects $C\in\mathrm{Ob}(\mathcal{C})$ and arrows $C\to C'$ encoding admissible changes of openness or external conditions;
  \item a Q-layer: for each $C\in\mathcal{C}$, a GKLS generator $\mathcal{L}_C$ on a fixed operator algebra together with a unique stationary state $\rho_C^\circ$ and a Doeblin-certified spectral gap $g(C)>0$;
  \item an S-layer: for each $C\in\mathcal{C}$, a self-preservation functional $\mathcal{F}_C$ on a space of internal charges $q(C)$ derived from pointer observables and stationary states, whose Hessian defines a positive-definite bilinear form on variations $dq(C)$.
\end{itemize}
Infinitesimal changes of openness $dN$ along the N-layer induce changes
$d\mathcal{L}_C$, $d\rho_C^\circ$, $dq(C)$, and $d^2\mathcal{F}_C$ along the Q- and S-layers, yielding the sensitivity chain above.
\end{definition}

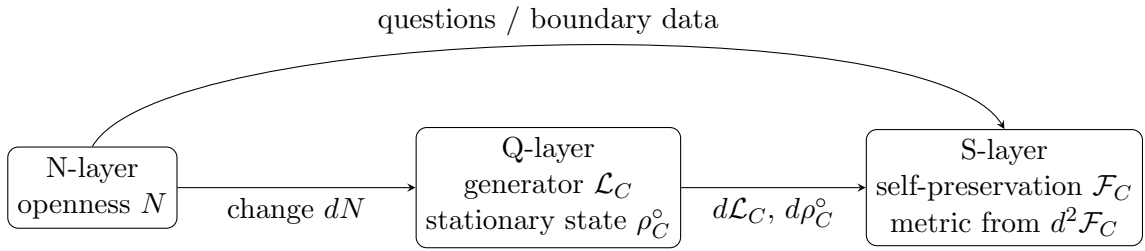
\begin{figure}[H]
\centering
\begin{tikzpicture}[node distance=6.0cm, >=stealth]
\node[draw, rounded corners, align=center] (N)
  {N-layer\\ openness $N$};
\node[draw, rounded corners, align=center, right of=N] (Q)
  {Q-layer\\ generator $\mathcal{L}_C$\\ stationary state $\rho_C^\circ$};
\node[draw, rounded corners, align=center, right of=Q] (S)
  {S-layer\\ self-preservation $\mathcal{F}_C$\\ metric from $d^2\mathcal{F}_C$};

\draw[->] (N) -- node[below]{change $dN$} (Q);
\draw[->] (Q) -- node[below]{$d\mathcal{L}_C,\, d\rho_C^\circ$} (S);

\draw[->, looseness=0.4, out=60, in=120]
  (N.north) to node[above]{questions / boundary data} (S.north);
\end{tikzpicture}
\caption{Schematic N--Q--S architecture and sensitivity chain
$dN \to d\mathcal{L}_C \to d\rho_C^\circ \to dq(C) \to d^2\mathcal{F}_C$.}
\label{fig:NQS-architecture}
\end{figure}

In Section~\ref{sec:Q-layer} we analyse the Q-layer and derive spectral gaps and Poisson-type sensitivity bounds from Doeblin minorization.
Section~\ref{sec:S-layer} constructs the S-layer self-preservation functionals and their Hessians, interpreted as intrinsic metrics on charge spaces.
Section~\ref{sec:N-layer} introduces a context graph, ghost couplings, and a global N--Q--S action, leading to discrete self-consistency equations and, in Section~\ref{sec:curvature}, to an N--Q--S connection and curvature.
Sections~\ref{sec:relation} and~\ref{sec:examples} relate these structures to classical Riemannian and information geometry and illustrate them on finite examples.

\subsection*{Riemannian geometry as a reference template}
For comparison, recall that classical Riemannian geometry specifies
a set of points $M$, tangent spaces $T_xM$, a metric $g$, a connection,
and the corresponding curvature tensor.
We will mirror this structure in an intrinsically operational way,
in the spirit of operational and information-theoretic formulations
of quantum theory and gravity.\cite{fuchsperes2000,hardy_op_gr,hoehnwever2018,delahamettegalley2024}

\subsection*{Main structural result}

The following informal theorem summarizes the main structural content of this paper; precise statements and proofs are given in Sections~\ref{sec:Q-layer}--\ref{sec:curvature}.

\begin{theorem}[Main structural result, informal]\label{thm:main-struct}
Assume an N--Q--S architecture as in Definition~\ref{def:NQS-architecture} with a finite context graph, GKLS generators with Doeblin-certified spectral gaps, and smooth self-preservation functionals.
Then:
\begin{enumerate}[label=(\roman*),leftmargin=2.5em]
  \item for each context $C$ the Hessian of $\mathcal{F}_C$ defines a positive-definite intrinsic metric on a space of charges, which in suitable parametric settings coincides with a quantum Fisher-type metric for the family of stationary states (Section~\ref{sec:S-layer});
  \item the Q-layer spectral gaps control the sensitivity of stationary states and induce a basic response inequality that bounds changes in $\rho_C^\circ$ in terms of perturbations of the generator $\mathcal{L}_C$ (Section~\ref{sec:Q-layer});
  \item a global N--Q--S action on context fields yields discrete self-consistency equations and, together with the sensitivity chain, an N--Q--S connection and curvature on the context graph, whose classical limits recover standard notions of metric, connection, and curvature on emergent manifolds and Fisher-type information geometry (Sections~\ref{sec:N-layer}--\ref{sec:relation}).
\end{enumerate}
\end{theorem}

We now turn to the detailed construction of each layer.

\section{Context space and pointer algebras}
This section introduces the basic objects of NQS-geometry:
contexts as ``points'', pointer algebras and internal symmetries
as the fibres over those points.

\begin{definition}[Context space]
A \emph{context space} is a small category or graph $\mathcal{C}$
whose objects $C \in \mathrm{Ob}(\mathcal{C})$ are called \emph{contexts}
and whose morphisms $C \to C'$ encode admissible transitions
of experimental or existential situations.
In the simplest setting, $\mathcal{C}$ may be regarded as a discrete set
equipped with a directed graph structure.
\end{definition}

\begin{definition}[Pointer algebra and internal symmetry]
To each context $C \in \mathcal{C}$ we associate
\begin{itemize}[leftmargin=2em]
  \item a finite-dimensional $C^\ast$-algebra $\mathcal{A}_C$ of \emph{pointer observables},
  \item a compact Lie group $G_{\mathrm{int}}(C)$ of \emph{internal symmetries}
        acting by $^\ast$-automorphisms on $\mathcal{A}_C$,
  \item a family of mutually commuting self-adjoint elements
        $\{H_a(C)\}_a \subset \mathcal{A}_C$ playing the role of Cartan generators.
\end{itemize}
We write $q_a(C) := \Tr(\rho_C^\circ H_a(C))$ for the corresponding \emph{Cartan charges}
in the stationary state $\rho_C^\circ$ (defined below).
\end{definition}

\begin{remark}
In physical examples, $\mathcal{A}_C$ is typically a matrix algebra
generated by coarse-grained pointer projectors, and
$G_{\mathrm{int}}(C)$ may contain factors such as
$\mathrm{SU}(3)_\mathrm{color}$, $\mathrm{SU}(2)_L$, or $\mathrm{U}(1)_Y$.
\end{remark}

\section{Q-layer: channels, fixed points, and mixing}
\label{sec:Q-layer}

This section formalizes the Q-layer as a family of primitive channels
and their instantaneous fixed points, together with a quantitative
notion of mixing or classicalization. The main purpose is to make
precise how small changes in openness $N$ propagate to changes in
channels and stationary states.

\subsection{Primitive channels and stationary states}

\begin{definition}[NQS channel family]
For each context $C \in \mathcal{C}$ we are given
a one-parameter family of completely positive trace-preserving maps
\[
  \Phi_t^{(C)} : \mathcal{A}_C \to \mathcal{A}_C, \qquad t \ge 0,
\]
forming a semigroup $\Phi_{t+s}^{(C)} = \Phi_t^{(C)} \circ \Phi_s^{(C)}$,
and admitting a unique faithful stationary state
$\rho_C^\circ$ such that $(\Phi_t^{(C)})^\ast(\rho_C^\circ) = \rho_C^\circ$
for all $t \ge 0$.
We write
\[
  \mathcal{L}_C := \left.\frac{\mathrm{d}}{\mathrm{d}t}\right|_{t=0} \Phi_t^{(C)}
\]
for the corresponding (GKLS) generator.
\end{definition}

\begin{definition}[Primitive channel]
We say that $\Phi_t^{(C)}$ (or equivalently $\mathcal{L}_C$) is
\emph{primitive} if for any two normal states $\rho,\sigma$ on
$\mathcal{A}_C$ we have
\[
  \lim_{t\to\infty} \,
  \bigl\|(\Phi_t^{(C)})^\ast(\rho) - (\Phi_t^{(C)})^\ast(\sigma)\bigr\|_1 = 0,
\]
and the stationary state $\rho_C^\circ$ is faithful.
\end{definition}

Primitivity ensures that the long-time behaviour is effectively
one-dimensional (everything collapses onto $\rho_C^\circ$),
which will later allow us to bound the response of $\rho_C^\circ$
to perturbations of the channel.

\subsection{Doeblin-type mixing and spectral gap}

\begin{definition}[Doeblin-type minorization]
We say that the semigroup $\Phi_t^{(C)}$ satisfies a
\emph{Doeblin-type minorization} if there exist $t_0>0$ and
$\varepsilon(C)>0$ such that for all normal states $\rho$,
\begin{equation}
  (\Phi_{t_0}^{(C)})^\ast(\rho) \;\ge\;
  \varepsilon(C)\,\rho_C^\circ.
  \label{eq:doeblin}
\end{equation}
\end{definition}

The lower bound~\eqref{eq:doeblin} means that after time $t_0$
the evolved state $(\Phi_{t_0}^{(C)})^\ast(\rho)$ always has a
uniform component $\varepsilon(C)\,\rho_C^\circ$ that is independent
of the initial state.

\begin{proposition}[Spectral gap from Doeblin-type minorization]
If $\Phi_t^{(C)}$ satisfies the minorization~\eqref{eq:doeblin},
then there exists a constant $g(C)>0$ such that for all $t\ge 0$,
\begin{equation}
  \bigl\|(\Phi_t^{(C)})^\ast(\rho) - \rho_C^\circ\bigr\|_1
  \;\le\; e^{-g(C)t} \,
  \bigl\|\rho - \rho_C^\circ\bigr\|_1
  \label{eq:gap-bound}
\end{equation}
for all normal states $\rho$.
The constant $g(C)$ may be interpreted as a
\emph{mixing rate} or \emph{spectral gap} at context $C$.
\end{proposition}

\begin{remark}
The estimate~\eqref{eq:gap-bound} can be seen as a non-commutative
version of the standard Doeblin--Dobrushin contraction bound
for classical Markov chains: the minorization~\eqref{eq:doeblin}
implies a contraction in trace norm, and hence a strictly
positive gap.
\end{remark}

\subsection{Poisson-type equation and sensitivity of fixed points}

We now recall the Poisson-type equation that links perturbations
of the channel to perturbations of the stationary state.

\begin{definition}[Sensitivity triple and Poisson equation]
Let $C$ be a fixed context, and consider a small perturbation of the
generator $\mathcal{L}_C \mapsto \mathcal{L}_C + \delta\mathcal{L}_C$.
We denote the corresponding first-order change in the stationary state
by $\delta\rho_C^\circ$. The pair $(\delta\mathcal{L}_C,\delta\rho_C^\circ)$
is said to satisfy the \emph{Poisson-type equation} if
\begin{equation}
  \mathcal{L}_C^\ast(\delta\rho_C^\circ)
  \;=\; -\,\delta\mathcal{L}_C^\ast(\rho_C^\circ),
  \label{eq:poisson-Q}
\end{equation}
together with the normalization condition $\Tr(\delta\rho_C^\circ)=0$.
\end{definition}

Formally,~\eqref{eq:poisson-Q} is obtained by differentiating the
stationarity condition $\mathcal{L}_C^\ast(\rho_C^\circ)=0$
with respect to the perturbation parameter.

\begin{proposition}[Gap-controlled sensitivity bound]
\label{prop:gap-sensitivity}
Assume that $\Phi_t^{(C)}$ is primitive and satisfies the
gap estimate~\eqref{eq:gap-bound} with constant $g(C)>0$.
Let $(\delta\mathcal{L}_C,\delta\rho_C^\circ)$ solve
the Poisson equation~\eqref{eq:poisson-Q}.
Then
\begin{equation}
  \|\delta\rho_C^\circ\|_1
  \;\le\;
  \frac{1}{g(C)}\,
  \bigl\|\delta\mathcal{L}_C^\ast(\rho_C^\circ)\bigr\|_1.
  \label{eq:basic-sensitivity}
\end{equation}
More generally, if we measure perturbations of the generator
in an operator norm $\|\cdot\|_{1\to 1}$, we obtain
\begin{equation}
  \|\delta\rho_C^\circ\|_1
  \;\le\;
  \frac{1}{g(C)}\,
  \|\delta\mathcal{L}_C\|_{1\to 1}.
  \label{eq:operator-sensitivity}
\end{equation}
\end{proposition}

\begin{proof}
Restrict $\mathcal{L}_C^\ast$ to the subspace
$\mathcal{B}_0 := \{ X \mid \Tr(X)=0\}$ of traceless operators.
The Doeblin-type minorization and the mixing estimate
\eqref{eq:gap-bound} imply that the dual semigroup
$e^{t\mathcal{L}_C^\ast}$ is a strict contraction on $\mathcal{B}_0$
in trace norm, and hence that the spectrum of $\mathcal{L}_C^\ast$
on $\mathcal{B}_0$ is contained in $\{\lambda : \Re\lambda\le -g(C)\}$;
see e.g.\ Davies~\cite{davies_open_systems}, Chap.~2.
It follows that the inverse of $-\mathcal{L}_C^\ast$ on $\mathcal{B}_0$
is bounded by
\[
  \|(-\mathcal{L}_C^\ast)^{-1}\|_{1\to 1} \le \frac{1}{g(C)}.
\]
Writing the Poisson equation as
$\mathcal{L}_C^\ast(\delta\rho_C^\circ)
 = -\delta\mathcal{L}_C^\ast(\rho_C^\circ)$
and solving for $\delta\rho_C^\circ$ then gives
\[
  \|\delta\rho_C^\circ\|_1
  \le \frac{1}{g(C)}\,
      \bigl\|\delta\mathcal{L}_C^\ast(\rho_C^\circ)\bigr\|_1,
\]
which is \eqref{eq:basic-sensitivity}. The operator-norm version
\eqref{eq:operator-sensitivity} follows immediately.
\end{proof}

The inequality~\eqref{eq:operator-sensitivity} will be one of the key
inputs when we relate channel perturbations, charge perturbations,
and the intrinsic metric in the next section.

\section{S-layer: self-preservation functionals and intrinsic metric}
\label{sec:S-layer}

This section introduces self-preservation functionals $\mathcal{F}_C$
on internal charge space, and defines an intrinsic metric as their Hessian
at self-consistent charges. We also relate this metric to Fisher-type
information metrics in simple examples.

\subsection{Self-preservation functional and self-consistency}

\begin{definition}[Internal charge space]
For each context $C$ we fix a finite family of commuting Cartan
generators $\{H_a(C)\}_{a=1}^{r_C} \subset \mathcal{A}_C$ and
define the associated charge map
\[
  q_a(C;\rho) := \Tr(\rho\,H_a(C)), \qquad a=1,\dots,r_C,
\]
for any state $\rho$ on $\mathcal{A}_C$.
We write $q(C;\rho) = (q_1(C;\rho),\dots,q_{r_C}(C;\rho)) \in \mathbb{R}^{r_C}$.
\end{definition}

\begin{definition}[Self-preservation functional]
For each context $C$, a \emph{self-preservation functional} is a function
\[
  \mathcal{F}_C : \mathbb{R}^{r_C} \times \prod_{C'\sim C}\mathbb{R}^{r_{C'}}
  \to \mathbb{R},
  \qquad (q(C),\{q(C')\}_{C'\sim C}) \mapsto \mathcal{F}_C,
\]
which quantifies the cost of realizing charges $q(C)$ given the
neighbouring charges $\{q(C')\}_{C'\sim C}$.
We often write
\[
  \mathcal{F}_C(q) =
  V_{\mathrm{int}}(q(C)) + V_{\mathrm{neigh}}(q(C);\{q(C')\}_{C'\sim C}),
\]
where $V_{\mathrm{int}}$ penalizes deviations from internally preferred
charges, and $V_{\mathrm{neigh}}$ penalizes mismatches with neighbouring
contexts.
\end{definition}

\begin{definition}[Self-consistent charges]
A family of charges $\{q^\star(C)\}_{C\in\mathcal{C}}$ is called
\emph{self-consistent} if, for each $C$,
\[
  q^\star(C) \in \arg\min_{q(C)} \,
  \mathcal{F}_C\bigl(q(C),\{q^\star(C')\}_{C'\sim C}\bigr).
\]
We assume that such a family exists and is locally unique.
\end{definition}

In simple quadratic models one may take
\begin{equation}
  \mathcal{F}_C(q)
  = \frac{1}{2}\,(q(C)-\bar{q}(C))^\top K_C\,(q(C)-\bar{q}(C))
    + \frac{1}{2}\sum_{C'\sim C} w_{CC'}\,
      \|q(C)-q(C')\|^2,
  \label{eq:quadratic-F}
\end{equation}
with $K_C$ positive definite and $w_{CC'}\ge 0$, where $\bar{q}(C)$
is a preferred internal configuration.

\subsection{Intrinsic metric from the Hessian}

\begin{definition}[Intrinsic NQS-metric]
The \emph{intrinsic metric} at context $C$ is the symmetric matrix
\begin{equation}
  g_{ab}(C)
  := \left.
     \frac{\partial^2 \mathcal{F}_C}
          {\partial q_a(C)\,\partial q_b(C)}
     \right|_{q = q^\star},
  \label{eq:NQS-metric}
\end{equation}
where $q^\star$ denotes a self-consistent configuration.
For small perturbations $\delta q(C)$ around $q^\star(C)$, the
quadratic form
\begin{equation}
  \delta s^2
  := \sum_{a,b} g_{ab}(C)\,\delta q_a(C)\,\delta q_b(C)
\end{equation}
approximates the second-order increase in self-preservation cost
at context $C$.
\end{definition}

In the quadratic example~\eqref{eq:quadratic-F}, we simply have
$g_{ab}(C) = (K_C)_{ab}$, so the metric directly reflects the
internal stiffness of the self-preservation functional.

\subsection{Relation to Fisher information in simple models}

To connect the NQS-metric~\eqref{eq:NQS-metric} with Fisher-type
information metrics, we consider a parametric family of channels
and stationary states.

\begin{definition}[Parametric channel family]
Let $\theta = (\theta^1,\dots,\theta^k)$ be a set of parameters
controlling the generator $\mathcal{L}_C(\theta)$ and stationary
state $\rho_C^\circ(\theta)$. We denote derivatives by
$\partial_i := \partial/\partial\theta^i$.
\end{definition}

A natural choice of self-preservation functional in this setting is
\begin{equation}
  \mathcal{F}_C(\theta)
  := D\bigl(\rho_C^\circ(\theta)\,\big\|\,\rho_C^\circ(\theta^\star)\bigr),
  \label{eq:F-from-entropy}
\end{equation}
where $D(\rho\|\sigma)$ denotes the quantum relative entropy and
$\theta^\star$ is the self-consistent parameter value.

\begin{proposition}[Hessian and Fisher-type metric]
Assume that $\theta^\star$ is a local minimum of~\eqref{eq:F-from-entropy}
and that the family $\rho_C^\circ(\theta)$ is sufficiently smooth.
Then the Hessian
\[
  g_{ij}(C)
  := \left.\frac{\partial^2 \mathcal{F}_C(\theta)}
                 {\partial\theta^i\,\partial\theta^j}
      \right|_{\theta=\theta^\star}
\]
coincides with a quantum Fisher-type information metric associated
with the family $\rho_C^\circ(\theta)$, in the sense that there
exists a choice of symmetric logarithmic derivatives $L_i$ such that
\[
  g_{ij}(C) = \Tr\bigl(\rho_C^\circ(\theta^\star) L_i L_j\bigr).
\]
\end{proposition}

\begin{remark}
The precise form of the Fisher metric depends on the choice of
monotone metric on the state space; see, for example, Petz's
characterization of monotone metrics on matrix spaces\cite{petz1996}
and standard references on information geometry.\cite{amariNagaoka2000,oizumiTsuchiyaAmari2015}
The present construction selects a particular monotone metric
through the operational choice of self-preservation functional.
\end{remark}

\subsection{Combining gap and metric: a basic response inequality}

Finally we combine the gap-controlled sensitivity from
Proposition~\ref{prop:gap-sensitivity} with the intrinsic metric.

Suppose that the charges $q_a(C)$ are smooth functions of the
stationary state $\rho_C^\circ$, for example
$q_a(C) = \Tr(\rho_C^\circ H_a(C))$. A small change
$\delta\rho_C^\circ$ then induces a change
$\delta q(C)$ and hence a change in self-preservation cost
\[
  \delta^2 \mathcal{F}_C
  \approx \frac{1}{2}\,\delta q(C)^\top g(C)\,\delta q(C).
\]

Combining the Lipschitz continuity of $q(C)$ with the
sensitivity bound~\eqref{eq:basic-sensitivity}, we obtain
schematically
\begin{equation}
  \delta^2 \mathcal{F}_C
  \;\lesssim\;
  \frac{L_C^2}{2\,g(C)^2}\,
  \bigl\|\delta\mathcal{L}_C^\ast(\rho_C^\circ)\bigr\|_1^2,
\end{equation}
where $L_C$ is a Lipschitz constant relating $\delta\rho_C^\circ$
to $\delta q(C)$. Thus contexts with large spectral gap $g(C)$
are intrinsically less sensitive in the self-preservation metric:
they require larger generator perturbations to achieve the same
increase in cost.

\section{N-layer: context graph, ghost couplings, and Laplacians}
\label{sec:N-layer}

This section encodes how contexts are linked, how ghost roles
couple distant situations, and how discrete Laplacians and
global self-consistency conditions arise. At this level, the
N-layer provides the analog of connection between different
``points'' of NQS-geometry.

\subsection{Context graph and ghost couplings}

\begin{definition}[Context graph and weights]
The context space $\mathcal{C}$ is equipped with an undirected graph
structure: we write $C \sim C'$ if the two contexts are directly coupled.
To each edge $(C,C')$ we assign a symmetric non-negative weight
$w_{CC'} = w_{C'C}$ representing the strength of \emph{ghost coupling}
between $C$ and $C'$.
\end{definition}

Intuitively, $w_{CC'}$ measures how easily certain ghost roles can
be shared or transported between $C$ and $C'$ without breaking
self-preservation too strongly.

\begin{definition}[Context Laplacian]
The \emph{context Laplacian} $\Delta$ acts on functions
$f : \mathcal{C} \to \mathbb{R}^d$ by
\begin{equation}
  (\Delta f)(C)
  := \sum_{C' \sim C} w_{CC'}\,\bigl(f(C) - f(C')\bigr).
  \label{eq:context-Laplacian}
\end{equation}
\end{definition}

When $\mathcal{C}$ is large and admits a coarse-grained manifold
description, the operator $\Delta$ approximates a Laplace--Beltrami
operator, and in directed or time-oriented variants it may give rise
to d'Alembert-type operators.

\subsection{Global self-preservation action and discrete field equations}

We now regard the collection of charges $\{q(C)\}$ as a field on the
context graph, and we consider a global self-preservation action.

\begin{definition}[Global self-preservation action]
Given local self-preservation functionals $\mathcal{F}_C$, we define
the \emph{global action}
\begin{equation}
  \mathcal{S}[q]
  := \sum_{C\in\mathcal{C}} \mathcal{F}_C\bigl(q(C);\{q(C')\}_{C'\sim C}\bigr)
     + \frac{1}{2}\sum_{C\sim C'} w_{CC'}\,
       \|q(C) - q(C')\|^2,
  \label{eq:global-action}
\end{equation}
where $\|\cdot\|$ is a fixed inner product on each charge space.
\end{definition}

\begin{proposition}[Discrete self-consistency equations]
A configuration $q^\star$ is a stationary point of
$\mathcal{S}[q]$ if and only if, for each context $C$,
\begin{equation}
  \frac{\partial \mathcal{F}_C}{\partial q_a(C)}(q^\star)
  + \sum_{C'\sim C} w_{CC'} \bigl(q_a^\star(C) - q_a^\star(C')\bigr)
  = 0,
  \qquad a = 1,\dots,r_C.
  \label{eq:discrete-EL}
\end{equation}
Equivalently, in vector notation,
\begin{equation}
  \nabla_{q(C)} \mathcal{F}_C(q^\star)
  + (\Delta q^\star)(C) = 0.
  \label{eq:EL-compact}
\end{equation}
\end{proposition}

\begin{remark}
Equation~\eqref{eq:EL-compact} is the discrete analog of a balance
between local self-preservation forces (the gradient of $\mathcal{F}_C$)
and elastic forces generated by the context Laplacian. In coarse-grained
limits, it takes the form of elliptic or hyperbolic field equations
for the charge fields $q(x)$ on emergent manifolds.
Such relations resonate with programmes that derive gravitational
field equations from information-theoretic or entanglement-based
principles.\cite{blancoCasiniHungMyers2013,swingle2012,jacobson2016,caoCarroll2022}
\end{remark}

\subsection{Toward Einstein-type equations in continuum limits}

When the context graph can be approximated by a smooth manifold $M$
and the charge fields $q(C)$ by smooth fields $q(x)$ on $M$, we expect
a continuum limit of~\eqref{eq:EL-compact} of the schematic form
\begin{equation}
  \nabla_q \mathcal{F}(q(x))
  + \kappa\,\Box q(x) = 0,
  \label{eq:continuum-EL}
\end{equation}
where $\Box$ denotes a Laplace--Beltrami or d'Alembertian operator
with respect to an emergent metric $g_{\mu\nu}(x)$.

In the NQS-geometry perspective, the emergent metric itself is
constructed from the Hessian of $\mathcal{F}$ with respect to charges,
as in~\eqref{eq:NQS-metric}, and the Einstein-type equations for
$g_{\mu\nu}$ arise as consistency conditions linking
\begin{itemize}[leftmargin=2em]
  \item the curvature of this metric,
  \item the global self-preservation constraints~\eqref{eq:continuum-EL},
  \item and the flows of Cartan charges across context space.
\end{itemize}
A detailed derivation of such Einstein-type equations is beyond the
scope of this first paper, but~\eqref{eq:global-action}--\eqref{eq:EL-compact}
provide the discrete prototype of these relations.

\section{Questioning, transport, and NQS-curvature}
\label{sec:curvature}

This section defines an analog of connection and curvature,
based on how small questioning perturbations propagate
through N, Q, and S layers along context paths.

\subsection{Sensitivity triples as tangent objects}

A small change in openness $dN$ induces a change in channels $d\Phi$
and hence in stationary states $d\rho^\circ$. We treat the triple
$(dN, d\Phi, d\rho^\circ)$ as an intrinsic tangent-like object at a context.

\begin{definition}[Sensitivity triple]
At context $C$, a \emph{sensitivity triple} is a triple
$(dN, d\Phi^{(C)}, d\rho_C^\circ)$ such that
\begin{equation}
  \mathcal{L}_C^\ast(d\rho_C^\circ)
  = -\,d\mathcal{L}_C^\ast(\rho_C^\circ),
  \qquad d\mathcal{L}_C = \Xi_C(dN),
  \label{eq:sensitivity-triple}
\end{equation}
for some linear response map
$\Xi_C : \mathrm{T}_N \to \mathrm{End}(\mathcal{A}_C)$
encoding how changes in openness modify the generator.
We denote by $T_C$ the vector space of such triples at $C$.
\end{definition}

\begin{remark}
The map $\Xi_C$ may be interpreted as an ``operational susceptibility''
of the agent or system: it specifies which generator perturbations
are reachable by modulating openness $N$.
\end{remark}

\subsection{Connection along context paths}

\begin{definition}[NQS-connection]
An \emph{NQS-connection} is a prescription which, to each oriented edge
$C \to C'$ in the context graph, associates a linear map
\[
  \Gamma_{C\to C'} :
  T_C \to T_{C'},
\]
such that along a path $C_0 \to C_1 \to \cdots \to C_n$
we can transport sensitivity data by composition
$\Gamma_{C_{n-1}\to C_n} \cdots \Gamma_{C_0\to C_1}$.
\end{definition}

Operationally, $\Gamma_{C\to C'}$ is determined by how
a change in openness and generator at $C$ is perceived or implemented
when the agent moves to $C'$, together with the ghost couplings
between the two contexts.

\subsection{Curvature as questioning non-commutativity}

\begin{definition}[NQS-curvature]
Given a closed loop
$\gamma : C_0 \to C_1 \to \cdots \to C_{k} = C_0$
in the context graph, the \emph{NQS-curvature operator} at $C_0$
along $\gamma$ is defined by
\begin{equation}
  \mathcal{R}(\gamma)
  := \Gamma_{C_{k-1}\to C_k}\cdots\Gamma_{C_0\to C_1}
     - \mathbbm{1}_{T_{C_0}},
  \label{eq:NQS-curvature}
\end{equation}
where $\mathbbm{1}_{T_{C_0}}$ is the identity on $T_{C_0}$.
The failure of $\mathcal{R}(\gamma)$ to vanish measures the
non-commutativity of questioning along the loop.
\end{definition}

\begin{remark}
In continuum limits where $T_C$ can be identified with ordinary
tangent spaces and $\Gamma_{C\to C'}$ with parallel transport
along infinitesimal steps,~\eqref{eq:NQS-curvature} reduces to
the usual holonomy operator associated with a connection, and
its infinitesimal form yields a Riemann curvature tensor.
\end{remark}

\subsection{A qutrit example: holonomy between two contexts}

To illustrate NQS-curvature in a concrete setting, consider two contexts
$C_1$ and $C_2$ describing a single qutrit pointer with internal
$\mathrm{SU}(3)$ symmetry. For each $i=1,2$ we take
$\mathcal{A}_{C_i} = M_3(\mathbb{C})$ and choose Cartan generators
$H_3^{(i)},H_8^{(i)}$ corresponding to the Gell-Mann matrices
$\lambda_3,\lambda_8$. We denote by
$q^{(i)} = (q^{(i)}_3,q^{(i)}_8)$
the associated Cartan charges in the stationary state $\rho^\circ_{C_i}$.

We assume that the two contexts are coupled by a partial swap channel
\begin{equation}
  \Psi_\theta(\rho)
  = (1-\theta)\,\rho + \theta\,U_{\mathrm{swap}}\rho\,U_{\mathrm{swap}}^\dagger,
  \qquad 0 < \theta \ll 1,
  \label{eq:partial-swap}
\end{equation}
where $U_{\mathrm{swap}}$ exchanges a pair of basis states and can be
implemented as a fragment of an $\mathrm{SU}(3)$ rotation.
Operationally, $\Psi_\theta$ realizes a weak ghost coupling between $C_1$
and $C_2$.

A simple questioning loop $\gamma$ of the form
\begin{equation}
  C_1 \xrightarrow{\;\Psi_\theta\;} C_2
      \xrightarrow{\;\Psi_{-\theta}\;} C_1
  \label{eq:qutrit-loop}
\end{equation}
transports a sensitivity triple $(dN,d\Phi^{(1)},d\rho_{C_1}^\circ)$ from $C_1$
to $C_2$ and back again. At the level of charges, we may write
the induced parallel transport as
\[
  q^{(1)} \mapsto q^{(2)} = U(\theta)\,q^{(1)}, \qquad
  q^{(2)} \mapsto \tilde{q}^{(1)} = U(-\theta)\,q^{(2)},
\]
where $U(\theta)$ is a $2\times 2$ matrix describing how the Cartan
charges transform under the partial swap~\eqref{eq:partial-swap}.
Expanding to second order in $\theta$, we obtain
\begin{equation}
  \tilde{q}^{(1)} - q^{(1)}
  = \theta^2\,K\,q^{(1)} + O(\theta^3),
  \label{eq:qutrit-holonomy}
\end{equation}
for some $2\times 2$ matrix $K$ determined by the commutator structure
of the Cartan generators with the swap part of $U_{\mathrm{swap}}$.

\begin{proposition}[Qutrit holonomy as NQS-curvature]
In the above setting, the deviation
$\tilde{q}^{(1)} - q^{(1)}$ in~\eqref{eq:qutrit-holonomy}
is the action of the NQS-curvature operator $\mathcal{R}(\gamma)$,
restricted to charge space. In particular, $K$ encodes an
$\mathrm{SU}(3)$ holonomy phase associated with the loop~\eqref{eq:qutrit-loop}.
\end{proposition}

\begin{remark}
Although this example is deliberately simple, it shows how
\emph{holonomies of contextual transport} can induce measurable
shifts in internal charges, providing an operational handle on
NQS-curvature. In more elaborate models, similar loops may probe
holonomies associated with larger $\mathrm{SU}(3)$ subgroups or
with extended networks of contexts.
\end{remark}

\section{Relation to Riemannian and information geometry}
\label{sec:relation}

In this section we relate the intrinsic NQS-geometry introduced above
to more familiar structures from Riemannian geometry and information
geometry. The central idea is that
\begin{itemize}[leftmargin=2em]
  \item the Hessian metric $g_{ab}(C)$ generalizes Fisher-type
        information metrics on statistical manifolds,
  \item the context Laplacian $\Delta$ induces Laplace--Beltrami
        operators on coarse-grained manifolds,
  \item NQS-curvature reduces to Riemannian curvature when contexts
        parametrize ordinary space--time points.
\end{itemize}

\subsection{Classical Riemannian limit on charge manifolds}

We first consider a situation where the space of possible charges
$\mathcal{Q}_C \subset \mathbb{R}^{r_C}$ at a fixed context $C$
forms a smooth manifold, and the self-preservation functional
$\mathcal{F}_C$ depends only on $q(C)$.

\begin{definition}[Charge manifold]
Let $C$ be a fixed context and assume that the set of admissible
charges $\mathcal{Q}_C$ is a smooth submanifold of $\mathbb{R}^{r_C}$.
We regard $(\mathcal{Q}_C,g(C))$ as a Riemannian manifold, where
$g(C)$ is the intrinsic NQS-metric defined by~\eqref{eq:NQS-metric},
restricted to tangent vectors in $T_q\mathcal{Q}_C$.
\end{definition}

In local coordinates $\xi = (\xi^1,\dots,\xi^d)$ on $\mathcal{Q}_C$,
the metric components are
\[
  g_{ij}^{(\xi)}(q)
  = \frac{\partial q_a}{\partial \xi^i}\,
    g_{ab}(C)\,
    \frac{\partial q_b}{\partial \xi^j},
\]
where repeated indices are summed over.

\begin{proposition}[Geodesics as optimal self-preservation paths]
Small deformations of charges $q(\tau)$ on $\mathcal{Q}_C$, parametrized
by $\tau\in[0,1]$, incur a second-order self-preservation cost
\[
  \delta^2 \mathcal{F}_C
  \approx \frac{1}{2}\int_0^1
    \dot{q}(\tau)^\top g(C)\,\dot{q}(\tau)\,\mathrm{d}\tau.
\]
In this approximation, curves that minimize $\delta^2\mathcal{F}_C$
for fixed endpoints are precisely the geodesics of the Riemannian
manifold $(\mathcal{Q}_C,g(C))$.
\end{proposition}

\begin{remark}
Thus the intrinsic metric $g(C)$ not only measures local sensitivity
but also determines the preferred interpolation paths between
charge configurations, in analogy with geodesics in classical
Riemannian geometry.
\end{remark}

\subsection{Information-geometric limit: Fisher and monotone metrics}

We now specialize to cases where charges $q(C)$ are in one-to-one
correspondence with parameters $\theta$ of a family of stationary
states $\rho_C^\circ(\theta)$, and the self-preservation functional
is derived from an information-theoretic divergence.

\begin{definition}[Information-geometric self-preservation]
Let $\theta \mapsto \rho_C^\circ(\theta)$ be a smooth family of states.
We define
\begin{equation}
  \mathcal{F}_C(\theta)
  := D\bigl(\rho_C^\circ(\theta)\,\big\|\,\rho_C^\circ(\theta^\star)\bigr),
  \label{eq:info-F}
\end{equation}
where $D(\rho\|\sigma)$ is a quantum divergence (e.g.\ relative entropy)
and $\theta^\star$ is a self-consistent parameter.
\end{definition}

In this setting, the Hessian metric
\[
  g_{ij}(C)
  = \left.
    \frac{\partial^2 \mathcal{F}_C(\theta)}
         {\partial\theta^i\,\partial\theta^j}
    \right|_{\theta=\theta^\star}
\]
is a quantum information metric. For relative entropy and commuting
families, it reduces to the classical Fisher information metric.

\begin{proposition}[Classical Fisher limit]
Assume that the states $\rho_C^\circ(\theta)$ commute for all $\theta$,
so that we can write
\[
  \rho_C^\circ(\theta)
  = \sum_x p_\theta(x)\,\dyad{x}
\]
in a fixed orthonormal basis.
If $\mathcal{F}_C(\theta)$ is given by~\eqref{eq:info-F}
with quantum relative entropy, then
\[
  g_{ij}(C) =
  \sum_x p_{\theta^\star}(x)\,
  \partial_i\log p_{\theta}(x)\big|_{\theta^\star}\,
  \partial_j\log p_{\theta}(x)\big|_{\theta^\star},
\]
the classical Fisher information matrix of the family $p_\theta$.
\end{proposition}

In the fully quantum case, different choices of divergence $D$
select different monotone metrics in the sense of Petz. The
self-preservation principle thus acts as an \emph{operational
selection principle} for information metrics.

\begin{remark}
From this viewpoint, NQS-geometry offers a unifying language in which
Fisher metrics, Bures--Helstrom metrics, and Fubini--Study metrics
appear as special cases corresponding to particular self-preservation
functionals and charge parametrizations.
\end{remark}

\subsection{Emergent Laplace--Beltrami and d'Alembert operators}

We now consider the global context space $\mathcal{C}$ and its
coarse-grained manifold limits.

Assume that $\mathcal{C}$ can be approximated by a smooth manifold $M$
and that the weights $w_{CC'}$ are such that the context Laplacian
$\Delta$ converges to a second-order differential operator on $M$.
Under mild regularity conditions, this limit operator is a
Laplace--Beltrami operator associated with an effective metric
$g_{\mu\nu}(x)$ on $M$.

In time-oriented settings where edges carry directions and are
interpreted as causal relations, an analogous limit leads to
d'Alembert-type operators $\Box_g$ on Lorentzian manifolds.
In both cases, the local NQS-metric $g_{ab}(C)$ on charge space
and the global operator $\Delta$ (or $\Box_g$) on context space
jointly determine the effective dynamics of charge fields $q(x)$.

\subsection{NQS-curvature and Riemann curvature}

Finally, when contexts are labelled by points $x\in M$ of an emergent
manifold, and the sensitivity space $T_C$ can be identified with the
tangent space $T_xM$, the NQS-connection $\Gamma_{C\to C'}$ becomes
a parallel transport operator along curves in $M$, and the NQS-curvature
$\mathcal{R}(\gamma)$ reduces to the usual holonomy of a connection.

In the infinitesimal limit, the operator-valued curvature
$\mathcal{R}(\gamma)$ defines a tensor
$R^\mu{}_{\nu\alpha\beta}$ on $M$ via
\[
  \mathcal{R}(\gamma_{\alpha\beta}) v^\mu
  \approx R^\mu{}_{\nu\alpha\beta} v^\nu,
\]
where $\gamma_{\alpha\beta}$ is a small loop in the $(\alpha,\beta)$
plane and $v$ is a tangent vector. In this way, the abstract NQS-curvature
captures and generalizes the familiar Riemann curvature tensor.

\begin{remark}
From the NQS perspective, Riemann curvature is one manifestation of a
more general notion of questioning non-commutativity: it is what remains
when we restrict attention to contexts that are already organized as
space--time points and to sensitivity triples that probe only geometric
degrees of freedom.
\end{remark}

\section{Simple illustrative examples}
\label{sec:examples}

We now present two minimal examples in which the abstract definitions
of NQS-geometry can be computed explicitly:
a finite Markov chain as a one-dimensional context line, and
a qutrit pointer with $\mathrm{SU}(3)$ internal symmetry.

\subsection{Finite Markov chain as a context line}

Consider a chain of three contexts $\mathcal{C} = \{C_1,C_2,C_3\}$
with nearest-neighbour couplings,
\[
  C_1 \sim C_2 \sim C_3, \qquad C_1 \not\sim C_3.
\]
We assign weights $w_{12}=w_{21}=w_{23}=w_{32}=w>0$ and
$w_{13}=w_{31}=0$.

\paragraph{Q-layer.}
At each context $C_i$ we consider a classical two-state system
with probabilities $(p_i,1-p_i)$ and a simple symmetric Markov
generator
\[
  \mathcal{L}_{C_i}(\rho)
  = \kappa_i \bigl( \rho^\circ_{C_i} - \rho \bigr),
\]
where $\rho^\circ_{C_i}$ is the stationary distribution and
$\kappa_i>0$ is a relaxation rate. The spectral gap is
$g(C_i)=\kappa_i$.

\paragraph{S-layer.}
We take a single charge $q(C_i)=p_i$ at each context and define
\[
  \mathcal{F}_{C_i}(p_i;\{p_j\}_{j\sim i})
  = \frac{1}{2}\,k_i\,(p_i-\bar{p}_i)^2
    + \frac{\lambda}{2}\sum_{j\sim i} (p_i-p_j)^2,
\]
with stiffness $k_i>0$, preferred value $\bar{p}_i$, and coupling
strength $\lambda\ge 0$.

The intrinsic metric at $C_i$ is simply
\[
  g(C_i) = \left.\frac{\partial^2\mathcal{F}_{C_i}}{\partial p_i^2}
          \right|_{p=\bar{p}}
         = k_i + \lambda\,\deg(C_i),
\]
where $\deg(C_1)=\deg(C_3)=1$ and $\deg(C_2)=2$.

\paragraph{N-layer.}
The context Laplacian~\eqref{eq:context-Laplacian} acts on
$f:\{1,2,3\}\to\mathbb{R}$ as
\[
  (\Delta f)(1) = w\,(f(1)-f(2)),\quad
  (\Delta f)(2) = w\,(2f(2)-f(1)-f(3)),\quad
  (\Delta f)(3) = w\,(f(3)-f(2)).
\]

The global action~\eqref{eq:global-action} becomes
\[
  \mathcal{S}[p]
  = \sum_{i=1}^3 \left[
      \frac{1}{2}k_i (p_i-\bar{p}_i)^2
    \right]
    + \frac{w+\lambda}{2}\bigl[(p_1-p_2)^2 + (p_2-p_3)^2\bigr].
\]

Stationarity $\delta\mathcal{S}/\delta p_i = 0$ yields the discrete
self-consistency equations
\begin{align*}
  &k_1(p_1-\bar{p}_1) + (w+\lambda)(p_1-p_2) = 0, \\
  &k_2(p_2-\bar{p}_2) + (w+\lambda)\bigl(2p_2 - p_1 - p_3\bigr) = 0, \\
  &k_3(p_3-\bar{p}_3) + (w+\lambda)(p_3-p_2) = 0.
\end{align*}

\begin{remark}
In the limit of many contexts arranged on a line or lattice,
with appropriate rescaling of $w+\lambda$, these equations converge
to a second-order differential equation of the form
\[
  k(x)\bigl(p(x)-\bar{p}(x)\bigr) + \kappa\,\partial_x^2 p(x) = 0,
\]
illustrating how NQS-geometry reproduces familiar diffusion or
elasticity equations in coarse-grained limits.
\end{remark}

\subsection{A qutrit pointer with $\mathrm{SU}(3)$ internal symmetry}

We revisit the qutrit setting from Section~\ref{sec:curvature} but
now focus on the intrinsic metric on Cartan charge space.

\paragraph{Setup.}
Let $\mathcal{A}_C = M_3(\mathbb{C})$ and choose Cartan generators
$H_3 = \lambda_3$, $H_8 = \lambda_8$, where $\lambda_3,\lambda_8$
are the standard Gell-Mann matrices. For a state $\rho$ we define
\[
  q_3 = \Tr(\rho H_3),\qquad
  q_8 = \Tr(\rho H_8),
\]
and write $q = (q_3,q_8)$.

\paragraph{Self-preservation functional.}
We take a quadratic self-preservation functional
\[
  \mathcal{F}_C(q)
  = \frac{1}{2}\,(q-q^\star)^\top K\,(q-q^\star),
\]
with positive definite $2\times 2$ matrix $K$ and preferred charges
$q^\star$ determined by the stationary state $\rho_C^\circ$.

The intrinsic metric at $C$ is then
\[
  g_{ab}(C) = K_{ab},
\]
independent of $q$.

\paragraph{Information-geometric interpretation.}
Suppose now that $\rho$ belongs to an $\mathrm{SU}(3)$-covariant
family of Gibbs states
\[
  \rho(\beta,h)
  = \frac{1}{Z(\beta,h)}\exp\bigl(-\beta(H_3 h_3 + H_8 h_8)\bigr),
\]
with inverse temperature $\beta$ and external ``fields'' $h=(h_3,h_8)$.
The charges $q_a(\beta,h)$ are expectation values of $H_a$.

Choosing the self-preservation functional as
\[
  \mathcal{F}_C(h)
  = D\bigl(\rho(\beta,h)\,\big\|\,\rho(\beta,h^\star)\bigr),
\]
with quantum relative entropy, we obtain at $h=h^\star$ the Hessian
metric
\[
  g_{ij}(C)
  = \frac{\partial^2\mathcal{F}_C}{\partial h_i\,\partial h_j}
  = \mathrm{Cov}_{\rho(\beta,h^\star)}(H_i,H_j),
\]
where the covariance is taken with respect to $\rho(\beta,h^\star)$.
Thus the intrinsic NQS-metric coincides with the quantum Fisher
metric associated with the Gibbs family.

\begin{remark}
In this example, the curvature effects discussed in
Section~\ref{sec:curvature} arise when we consider multiple contexts
$C_i$ with different local Hamiltonians and ghost couplings between
them. The local metric $g_{ij}(C_i)$ is the covariance matrix of
Cartan generators, while the NQS-curvature encodes the holonomy of
their transport along context loops.
\end{remark}

\section{Discussion and outlook}
\label{sec:discussion}

In this final section we summarize the main ingredients of
N--Q--S intrinsic geometry, discuss how familiar space--time and
gauge structures may appear as particular charts, and sketch several
directions for future work.

\subsection{From operational Riemannian to emergent Lorentzian}

Although the intrinsic metric $g_{ab}(C)$ introduced here is
Riemannian at the operational level, directed ghost couplings
on the context graph naturally distinguish forward and backward
directions. In coarse-grained regimes where contexts organize
into large-scale chains with a stable notion of ``before'' and ``after'',
this induces an effective time orientation and, under suitable
scaling limits, a Lorentzian signature on emergent manifolds.
This is conceptually similar to how causal set theory recovers
Lorentzian space--times from a partially ordered set endowed
with a volume measure.\cite{bombelli1987}

From the NQS perspective, the basic ingredients are:
\begin{itemize}[leftmargin=2em]
  \item contexts $C$ as vertices in a (possibly directed) graph,
  \item ghost couplings $w_{CC'}$ as edge weights encoding
        how easily ghost roles can be shared,
  \item self-preservation functionals $\mathcal{F}_C$ and their Hessians
        as local metrics on charge space,
  \item sensitivity triples and NQS-connections as rules for transporting
        questioning data along edges.
\end{itemize}
When directed edges are interpreted as causal relations, the combination
of these structures yields an effective Lorentzian geometry on a
coarse-grained manifold $M$, together with a notion of curvature
expressed as non-commutative questioning loops.

\subsection{Contexts, charges, and curvature: a unifying picture}

The constructions in this paper suggest the following unifying picture
of NQS-geometry:
\begin{itemize}[leftmargin=2em]
  \item \emph{Points} of the geometry are operational contexts
        equipped with pointer algebras and internal symmetries.
  \item \emph{Tangent directions} are sensitivity triples
        $(dN,d\Phi,d\rho^\circ)$ describing infinitesimal changes
        in openness, channels, and stationary states.
  \item \emph{Metric data} come from the Hessians of
        self-preservation functionals with respect to internal
        charges, leading to Fisher- and Fubini--Study-type metrics
        in appropriate limits.
  \item \emph{Connection} is given by contextual transport maps
        $\Gamma_{C\to C'}$ along edges of the context graph,
        combining ghost couplings with the agent's response map
        $\Xi_C$.
  \item \emph{Curvature} is encoded in the failure of questioning
        loops to commute, as captured by the NQS-curvature operator
        $\mathcal{R}(\gamma)$.
\end{itemize}

Classical Riemannian geometry is recovered when contexts are already
organized as points of a smooth manifold and when sensitivity triples
probe only geometric degrees of freedom. In that regime, the NQS-metric
reduces to a standard Riemannian metric $g_{\mu\nu}$, the context
Laplacian converges to a Laplace--Beltrami operator, and NQS-curvature
reduces to the Riemann curvature tensor.

\subsection{GR and gauge theories as charts of NQS-geometry}

A central motivation for NQS-geometry is the hypothesis that both
general relativity and the Standard Model can be seen as particular
charts or sectors of a more general context geometry.

On the gravitational side, the discrete self-consistency equations
\eqref{eq:EL-compact} suggest that, in continuum limits, charge fields
$q(x)$ obey elliptic or hyperbolic equations of the form
\eqref{eq:continuum-EL}. The emergent metric $g_{\mu\nu}(x)$ itself
is constructed from the Hessian of a global self-preservation action,
and Einstein-type equations arise as consistency conditions relating
\begin{itemize}[leftmargin=2em]
  \item the curvature of this metric,
  \item the flows of Cartan charges across context space,
  \item and the global minimization of self-preservation cost.
\end{itemize}
From this viewpoint, the Einstein equations are not postulated but
appear as balance equations in an intrinsically operational geometry.

On the gauge-theoretic side, the qutrit example of
Section~\ref{sec:curvature} illustrates how holonomies of contextual
transport naturally encode $\mathrm{SU}(3)$ phases on Cartan charges.
In more realistic models, one may consider families of contexts whose
pointer algebras support $\mathrm{SU}(3)\times\mathrm{SU}(2)\times
\mathrm{U}(1)$ internal symmetries. The proposal is that
\begin{itemize}[leftmargin=2em]
  \item gauge connections correspond to NQS-connections restricted
        to internal charge fibres,
  \item holonomies of these connections reproduce Wilson loops and
        other gauge-invariant observables,
  \item nucleons and other stable particles appear as strongly
        self-preserving condensates in context space, i.e.\ as
        minima of the global action~\eqref{eq:global-action}.
\end{itemize}
In this sense, gauge fields and matter content are not added ``by hand'',
but arise as structural features of transport and self-preservation in
NQS-geometry.

\subsection{Future directions and companion works}

The present article is intended as a first, foundational step.
Several directions for further development are particularly natural.

\paragraph{Gravitational charts in context geometry.}
A first companion paper would focus on regimes where the context space
admits a four-dimensional manifold approximation and where directed
ghost couplings induce a robust time orientation. The goal would be to
derive explicit Einstein-type equations for the emergent metric
from a global self-preservation principle and to compare them with
classical general relativity in simple cosmological or black-hole-like
setups.

\paragraph{Internal holonomies and $\mathrm{SU}(3)$ color.}
A second companion work would develop the internal fibre aspect of
NQS-geometry, with particular emphasis on $\mathrm{SU}(3)$ color.
Building on the qutrit example, one can study contextual networks in
which Cartan charges and ghost couplings are tuned to reproduce
hadronic spectra and confinement-like behaviour, interpreting color
holonomies as manifestations of NQS-curvature in internal directions.

\paragraph{Cognitive and psychological applications.}

Beyond physical applications, it is tempting to apply NQS-geometry to cognitive and social processes, where contexts represent experiential or group configurations and self-preservation functionals encode the stability of identities or narratives. In such settings, NQS-curvature would quantify the non-commutativity of questioning and the path dependence of meaning-making. We leave these more speculative directions to future work.

\paragraph{Mathematical refinements.}
On the mathematical side, it would be natural to:
\begin{itemize}[leftmargin=2em]
  \item classify self-preservation functionals that lead to known
        monotone metrics in quantum information theory,
  \item study the space of NQS-connections compatible with a given
        metric and Laplacian, in analogy with Levi-Civita connections,
  \item investigate the existence of NQS-analogs of geodesic completeness,
        singularities, and topological invariants.
\end{itemize}

\subsection{Concluding remarks}

We have argued that operational contexts, channels, and self-preservation
principles define a natural intrinsic geometry, in which metrics and
curvatures are reconstructed from sensitivity and questioning rather
than assumed from the outset. Ordinary Riemannian and Lorentzian
geometries, along with gauge fields, appear as particular regimes
of this broader structure.

Whether NQS-geometry ultimately provides a viable route to unifying
space--time, matter, and cognitive processes remains an open question.
However, even at the level developed here, it offers a concrete
framework in which ideas from quantum information, operational
foundations, and differential geometry can be brought into systematic
contact. Further elaboration of this framework, and confrontation
with both physical models and empirical data, will be the subject
of future work.

\section*{Statements and Declarations}

\paragraph{Competing interests}
The author has no relevant financial or non-financial interests to disclose.

\paragraph{Funding}
This research received no specific grant from any funding agency in the public,
commercial, or not-for-profit sectors.

\paragraph{Data availability}
No datasets were generated or analysed during the current study.

\paragraph{Use of AI-assisted tools}
Language editing and typesetting suggestions were partially assisted by an
AI system (ChatGPT, OpenAI). The author verified all content and takes full
responsibility for the text and for all mathematical claims.

\paragraph{Author contribution}
Kazuyuki Yoshida conceived the N--Q--S framework, developed the mathematical
formalism, and wrote the manuscript.

\bibliographystyle{plain}
\bibliography{nqs-geometry}
\end{document}